\documentclass[letterpaper,11pt]{article}
\usepackage[margin=1in]{geometry}  

\usepackage{ bbold }
\usepackage{listings}
\usepackage{amsmath}
\usepackage{amsthm}
\usepackage{tikz}
\usepackage{caption}
\usepackage{array}
\usepackage{mdwmath}
\usepackage{multirow}
\usepackage{mdwtab}
\usepackage{eqparbox}
\usepackage{multicol}
\usepackage{amsfonts}
\usepackage{tikz}
\usepackage{multirow,bigstrut,threeparttable}
\usepackage{amsthm}
\usepackage{array}
\usepackage{bbm}
\usepackage{subfigure}
\usepackage{epstopdf}
\usepackage{mdwmath}
\usepackage{mdwtab}
\usepackage{eqparbox}
\usepackage{tikz}
\usepackage{latexsym}
\usepackage{cite}
\usepackage{amssymb}
\usepackage{bm}
\usepackage{amssymb}
\usepackage{graphicx}
\usepackage{mathrsfs}
\usepackage{epsfig}
\usepackage{psfrag}
\usepackage{setspace}
\usepackage[
CJKbookmarks=true,
bookmarksnumbered=true,
bookmarksopen=true,
colorlinks=true,
citecolor=red,
linkcolor=blue,
anchorcolor=red,
urlcolor=blue
]{hyperref}
\usepackage{algorithm}
\usepackage{algpseudocode}
\usepackage{stfloats}

\theoremstyle{plain}
\newtheorem{theorem}{Theorem}

\newtheorem{lemma}{Lemma}

\newtheorem{corollary}{Corollary}
\newtheorem{definition}{Definition}

\newtheorem{question}{Question}
\newtheorem*{question*}{Question}

\usepackage{xspace}

\definecolor{myblue}{rgb}{.8, .8, 1}
\definecolor{mathblue}{rgb}{0.2472, 0.24, 0.6} 
\definecolor{mathred}{rgb}{0.6, 0.24, 0.442893}
\definecolor{mathyellow}{rgb}{0.6, 0.547014, 0.24}

\usepackage{cleveref}
\crefname{lemma}{Lemma}{Lemmas}
\Crefname{lemma}{Lemma}{Lemmas}
\crefname{thm}{Theorem}{Theorems}
\Crefname{thm}{Theorem}{Theorems}
\crefname{assumption}{Assumption}{Assumptions}
\Crefname{assumption}{Assumption}{Assumptions}
\crefformat{equation}{(#2#1#3)}

\usepackage{enumerate}

\def \cX {\mathcal{X}}

\begin{document}

\title{Minimax redundancy for Markov chains with large state space}

\author{Kedar Shriram Tatwawadi, Jiantao Jiao, Tsachy Weissman}

\date{\today}


%
\maketitle
\begin{abstract}
For any Markov source, there exist universal codes whose normalized codelength approaches the Shannon limit asymptotically as the number of samples goes to infinity. This paper investigates how fast the gap between the normalized codelength of the ``best'' universal compressor and the Shannon limit (i.e. the compression redundancy) vanishes non-asymptotically in terms of the alphabet size and mixing time of the Markov source. We show that, for Markov sources whose relaxation time is at least $1 + 
\frac{(2+c)}{\sqrt{k}}$, where $k$ is the state space size (and $c>0$ is a constant), the phase transition for the number of samples required to achieve vanishing compression redundancy is precisely $\Theta(k^2)$. 
\end{abstract}
%

\section{Introduction}
For any data source that can be modeled as a stationary ergodic stochastic process, it is well known in the literature of universal compression that there exist compression algorithms without any knowledge of the source distribution, such that its performance can approach the fundamental limit of the source, also known as the Shannon entropy, as the number of observations tends to infinity. The existence of universal data compressors has spurred a huge wave of research around it. A large fraction of practical lossless compressors are based on the Lempel--Ziv algorithms~\cite{ziv1977universal,
ziv1978compression} and their variants, and the normalized codelength of a universal source code is also widely used to measure the \emph{compressibility} of the source, which is based on the idea that the normalized codelength is ``\emph{close}'' to the true entropy rate given a moderate number of samples. 

There has been considerable efforts trying to quantify how fast the codelength of a universal code approaches the Shannon entropy rate. One of the general statements pertaining to distributions parametrized by a finite dimensional vector is due to Rissanen~\cite{rissanen1984universal}. Let $X^n$ be a sequence of random variables generated from some stationary distribution $p_{\theta}(x^n)$ with parameters $\theta$. A compressor $L$ for the $X^n$ sequence is characterized by its length function $L(x^n)$, which is the length (in bits), of the code corresponding to every realization $x^n$ of $X^n$. 

The entropy $H_\theta(X^n)$ quantifies the fundamental limit of compression under model $p_\theta$, which is given by
\begin{equation}
H_{\theta}(X^n) = \sum_{x^n} p_{\theta}(x^n) \log _2 \frac{1}{p_{\theta}(x^n)}
\end{equation}
\footnote{Throughout the paper, we will work with $\log \equiv \log_2$.}
The redundancy for a compressor with length function $L(X^n)$ is defined as:
\begin{equation}
r_n(L,\theta) = \frac{1}{n} \left( \mathbb{E}[L(X^n)|\theta] - H_{\theta}(X^n) \right)
\end{equation}
Rissanen~\cite{rissanen1984universal} states that if $\theta \in \Theta \subset \mathbb{R}^d$, and if the parameter $\theta$ can be estimated with ``\emph{parametric}'' rate asymptotically (with $d,\theta$ fixed), then there exists some compressor $L$ such that
\begin{align}
r_n(L,\theta)= \frac{d\log n}{2n} + O\left(\frac{1}{n}\right) 
\end{align}
as $n\to \infty$. Moreover, fixing $d,\epsilon>0$, for any uniquely decodable code $L$, its redundancy satisfies 
\begin{align} \label{eqn.rissanenlowerbound}
r_n(L,\theta)\geq (1-\epsilon) \frac{d\log n}{2n}
\end{align}
as $n\to \infty$ for \emph{all} values of $\theta$ except for a set whose volume vanishes as $n \to \infty$ while other parameters are fixed.

The focus of Rissanen~\cite{rissanen1984universal} was \emph{asymptotic}, i.e., the characterization of the redundancy as the number of samples $n\to \infty$ while other parameters remain fixed. There has been considerable generalizations in the asymptotic realm, such as~\cite{atteson1999asymptotic,merhav1995strong,Feder--Merhav1996,xie1997minimax,xie2000asymptotic}.

In modern applications, the parameter dimension $d$ may be comparable or even larger than the number of samples $n$. For example, in the Google 1 Billion Word dataset (1BW)~\cite{chelba2013one}, the number of distinct words is more than 2 million, and the data distribution is also not i.i.d., which makes us wonder whether we are operating in the asymptotics when any universal code is applied. We emphasize that the implications of~(\ref{eqn.rissanenlowerbound}) may not be \emph{correct} in the non-asymptotic setting (i.e. when the paprameter dimension $d$ is comparable to the number of samples $n$). Indeed, interpreting~(\ref{eqn.rissanenlowerbound}) in the non-asymptotic way, it implies that it requires at least $n \gg d \log d$ samples to achieve vanishing redundancy. However, when the data source is i.i.d. with alphabet size $d+1$, the precisely non-asymptotic computation shows that the phase transition between vanishing and non-vanishing redundancy is at $n \asymp d$~\cite{jiao2017estimating}.

There exists extensive literature on quantifying the redundancy in the non-asymptotic regime. Davisson~\cite{davisson1983minimax} considered the case of memoryless sources and $m$-Markov sources, and obtained non-asymptotic upper and lower bounds (i.e. bounds that are explicit in all the parameters involved) on the \emph{average case} minimax redundancy, which is defined by
\begin{align}\label{eqn.minimaxredundancy}
\inf_{L} \sup_{\theta \in \Theta} r_n(L,\theta),
\end{align}
where the infimum is taken over all uniquely decodable codes \cite{cover2012elements} (section 5.1). However, the lower bound for Markov sources with alphabet size $k$ in~\cite{davisson1983minimax} is non-zero only when $n \gg k^2\log k$  (See Appendix \ref{appendix.davisson}) and are not tight in the sense that the upper and lower bounds do not match in scaling in the large alphabet regime. The work~\cite{Orlitsky--Santhanam2004speaking,drmota2004precise,Szpankowski--Weinberger2012minimax} mainly considered a variant called \emph{worst case} minimax redundancy, and showed that for i.i.d. sources with alphabet size $k$, the worst case minimax redundancy~\footnote{Precisely, the minimax regret with respect to a coding oracle that only uses codes corresponding to i.i.d. distributions.} vanishes if and only if the number of samples $n \gg k$ non-asymptotically. The problem of worst-case minimax redundancy for Markov sources was considered in \cite{jacquet2004markov}.

The focus of this paper is on the average case minimax redundancy for Markov chains. We refine the minimax redundancy in~(\ref{eqn.minimaxredundancy}) and categorize different Markov chains by how fast it ``\emph{mixes}''. Informally, we ask the following question:
\begin{question}
How does the minimum number of samples required to achieving vanishing redundancy depend on the state space size and mixing time?
\end{question}


\section{Preliminaries}

Consider a first-order Markov chain $X_1,X_2,\ldots$ on a finite state space $\mathcal{X}= \{1,2,\ldots, k\} \triangleq [k]$ with transition kernel $K$. We denote the entries of $K$ as $K_{ij}$, that is, $K_{ij} = P_{X_2|X_1}(j|i)$ for $i, j \in \cX$. 
We say that a Markov chain is \emph{stationary} if $P_{X_1}$, the distribution of $X_1$, satisfies
\begin{equation}
\sum_{i = 1}^k P_{X_1} (i) K_{ij} = P_{X_1}(j) \qquad \text{for all }j \in \cX.
\end{equation}

We say that a Markov chain is \emph{reversible} if there exists a distribution $\pi$ on $\mathcal{X}$ which satisfies the detailed balance equations:
\begin{equation}
\pi_i K_{ij}  = \pi_j K_{ji} \qquad \text{for all }i,j \in \cX.
\end{equation}
In this case, $\pi$ is called the stationary distribution of the Markov chain. 

For a reversible Markov chain, its (left) spectrum of the operator $K$ consists of $k$ real eigenvalues $1 = \lambda_1 \geq \lambda_2 \geq \cdots \geq \lambda_k \geq -1$. We define the spectral gap of a reversible Markov chain as
\begin{equation}
\gamma(K) = 1- \lambda_2.
\end{equation}
The \emph{absolute spectral gap} of $K$ is defined as
\begin{equation}
\gamma^*(K) = 1 - \max_{i\geq 2} |\lambda_i|,
\end{equation}
and it clearly follows that, for any reversible Markov chain,
\begin{equation}
\gamma(K) \geq \gamma^*(K).
\end{equation}
The \emph{relaxation} time of a Markov chain is defined as
\begin{equation}
\tau_{\mathrm{rel}}(K) = \frac{1}{\gamma^*(K)}.
\end{equation}

The relaxation time of a reversible Markov chain (approximately) captures its mixing time, which informally is the smallest $n$ for which the marginal distribution of $X_n$ is very close to the Markov chain's stationary distribution. We refer to~\cite{montenegro2006mathematical} for a survey. Intuitively speaking, the shorter the relaxation time $\tau_\mathrm{rel}$, the faster the Markov chain ``mixes'': that is, the shorter its ``memory'', or the sooner evolutions of the Markov chain from different starting states begin to look similar.

The multiplicative reversibilization of the transition matrix $K$ is defined as:
\begin{equation}
K^{*}_{ji} = \frac{\pi_iK_{ij}}{\pi_j} 
\end{equation}
$K^{*}$ is infact the transition matrix for the reverse Markov chain $X_n\rightarrow X_{n-1} \rightarrow \ldots \rightarrow X_1$. Note that for reversible chains $K^{*} = K$. The pseudo-spectral gap for a non-reversible chain (with transition matrix $K$) is defined as:
\begin{equation}
\gamma_{ps}(K) = \max_{r \geq 1} \frac{\gamma((K^{*})^rK^r)}{r}
\end{equation}
The pseudo-spectral gap for a non-reversible chain is related to the mixing time of the non-reversible Markov chain. 

%



We denote by $\mathcal{M}_1(k)$ the set of all discrete distributions with alphabet size $k$ (\textit{i.e.}, the $(k-1)$-probability simplex), and by $\mathcal{M}_2(k)$ the set of all Markov chain transition matrices on a state space of size $k$. Let $\mathcal{M}_{2,\text{rev}}(k) \subset \mathcal{M}_2(k)$ be the set of transition matrices of all stationary \emph{reversible} Markov chains on a state space of size $k$. We define a class of stationary Markov chains $\mathcal{M}_{2,\text{rev}}(k, \tau_{\mathrm{rel}}) \subset \mathcal{M}_{2,\text{rev}}(k)$ as follows:
\begin{equation}
\mathcal{M}_{2,\text{rev}}(k, \tau_{\mathrm{rel}}) = \{ K_{ij} \in \mathcal{M}_{2,\text{rev}}(k), \tau_{\mathrm{rel}}(K) \leq \tau_{\mathrm{rel}}\}.
\end{equation}
In other words, we consider stationary reversible Markov chains whose relaxation time is upper-bounded by $\tau_{\mathrm{rel}}$. 

Another probabilistic representation of reversible Markov chains is via random walk on undirected graphs. Consider an undirected graph (without multi-edges) on $k$ vertices. Let the weight on the undirected edge $ \{i,j\}$ be denoted as $w_{ij} \geq 0$. Due to the undirected nature of the graph, $w_{ij} = w_{ji} \geq 0, \mbox{  } \forall i,j \in [k]$. We also define $\rho_i$ and $\rho$ as: 
\begin{align*}
\rho_i &= \sum_{j=1}^k w_{ij}, \forall i \in [k]\\
\rho &= \sum_{i,j} w_{ij}
\end{align*} 

Here, $\rho_i$ corresponds to the row-sums of the weight matrix $W$, with entries $w_{ij}$. We can now consider a Markov chain corresponding to a random walk on this graph. The transition probabilities and the stationary distribution corresponding to a random walk on this weighted undirected graph are given by:
\begin{align}
K_{ij} &= \frac{w_{ij}}{\rho_{i}} \\
\pi_i &= \frac{\rho_i}{\rho}
\end{align}

We can verify that the transition matrix $K$ corresponds to a reversible Markov chain (i.e. $K \in \mathcal{M}_{2,\text{rev}}(k)$) as:
\begin{align}
\pi_i K_{ij} &= \frac{w_{ij}}{\rho}\\
			 &= \frac{w_{ji}}{\rho}\\
			 &= \pi_j K_{ji}
\end{align}

Conversely, we can understand any reversible Markov chain $\hat{K} \in \mathcal{M}_{2,\text{rev}}(k)$, with stationary distribution $\hat{\pi}$ as a random walk on an undirected graph with weights $\hat{w}_{ij}$:

\begin{equation}
\hat{w}_{ij} = \hat{\pi}_i \hat{K}_{ij}
\end{equation}

The quantity of interest in this paper is
\begin{align}
R_n(k, \tau_{\mathrm{rel}}) = \inf_{L} \sup_{K \in \mathcal{M}_{2,\mathrm{rev}}(k,\tau_{\mathrm{rel}})} r_n(L,\theta),
\end{align}
where the infimum is taken over all uniquely decodable codes \cite{cover2012elements} (section 5.1), and the supremum is taken over all stationary reversible Markov chains whose relaxation time is upper bounded by $\tau_{\mathrm{rel}}$. We define the quantity $n^*(k, \tau_{\mathrm{rel}}, \epsilon)$ as follows:
\begin{align}
n^*(k, \tau_{\mathrm{rel}}, \epsilon) \triangleq \min\{n: R_{n}(k,\tau_{\mathrm{rel}})\leq \epsilon\}. 
\end{align}

\subsection*{Notation}
The quantity $h(X)$ denotes the differential entropy of the continuous random variable $X$ with density function $f_X(x)$, and is given by: 
\begin{equation}
h(X) = \int f_X(x) \log_2 \frac{1}{f_X(x)} dx
\end{equation}
We define the KL-divergence $D(p_X||q_X)$ between two discrete distributions $p_X(x)$ and $q_X(x)$ as:
\begin{equation}
D(p_X||q_X) = \sum_x p_X(x) \log_2 \frac{p_X(x)}{q_X(x)}
\end{equation}  
Thoughout the paper, we will use the notation $o(.)$ and $O(.)$ to denote the asymptotic growth of a function. Let $f(k)$ and $g(k)$ be non-negative functions. We say that function $f(k) = O(g(k))$, if $f(k) \leq Cg(k)$ for some $C >0$ and all $n > C$. We say that function $f(k) = o(g(k))$, if the asymptotic growth of $f(k)$ is strictly slower than that of $g(k)$, i.e.
\begin{equation*}
\lim_{k \rightarrow \infty} \frac{f(k)}{g(k)} = 0
\end{equation*}



\section{Main Results}

The main theorems in this paper are: 
\begin{theorem}\label{thm.main}
For $\tau_{\mathrm{rel}} \geq 1 + \frac{2+c}{\sqrt{k}} $, 
\begin{align}
R_n(k, \tau_{\mathrm{rel}}) & \geq \frac{k(k-1)}{4n}\log \frac{2(n-1)}{k(k-1)} + \frac{k(k-1)}{4n} \log \frac{e}{16\pi (1+ \frac{2+c}{\sqrt{k}})} - \frac{\log k}{n}		.
\end{align}
for $k\geq k_c$, where $c>0$ is a constant and $k_c$ is a constant depending only on $c$. 
\end{theorem}

Theorem \ref{thm.main} is proved in the section \ref{proof.thm.main}.  

\begin{theorem}\label{thm.main2}
The average-case minimax redundancy $R_n(k)$ for Markov sources is defined as:
\begin{align*}
R_n(k) = \inf_{L} \sup_{\theta \in \mathcal{M}_{2}(k)} r_n(L,\theta)
\end{align*}
Then, the following upper bound holds:
\begin{align*}
R_n(k) &\leq \frac{2k^2}{n}\log_2 \left( \frac{n}{k^2}  + 1\right) + \frac{k^2}{n} + \frac{\log_2 k + 3}{n}
\end{align*}
\end{theorem}
Note that as $R_n(k) \geq R_n(k, \tau_{\mathrm{rel}})$, the upper bound in theorem \ref{thm.main2} is valid for $R_n(k, \tau_{\mathrm{rel}})$. Also, as $R_n(k) \geq R_n(k, \tau_{\mathrm{rel}})$, the lower bound in theorem \ref{thm.main} is valid for $R_n(k)$. Theorem \ref{thm.main2} is proved in the section \ref{proof.thm.main2}.  The following corollary is immediate.  

\begin{corollary}\label{corollary.main}
If $n\gg k^2$, then $R_{n}(k,\tau_{\mathrm{rel}})\to 0$ uniformly over $\tau_{\mathrm{rel}}$. 
For $\tau_{\mathrm{rel}} \geq 1 + \frac{2+c}{\sqrt{k}}$, where $c>0$ is a positive constant, there exists a constant $c_1$ such that if $n = c_1 k^2$, then $R_n(k, \tau_{\mathrm{rel}})$ is bounded away from zero as $k\to \infty$. 
\end{corollary}

Analyzing $R_{n}(k,\tau_{\mathrm{rel}})$ over reversible Markov chains, gives us a more refined understanding of the compression redundancy. From theorem \ref{thm.main2}, we observe that for any Markov distribution, we can achive $\epsilon$ redundancy (for any constant $\epsilon > 0$) using $n \propto k^2$ samples. On the other hand, theorem \ref{thm.main} tells us that, even for the small family of fast mixing reversbile chains, in the worst case, at least $O(k^2)$ samples are necessary to obtain $\epsilon$ redundancy.  
   
Figure~\ref{fig.phasetransition} provides a pictorial illustration of $n^*(k, \tau_{\mathrm{rel}}, \epsilon)$ when $\epsilon$ is a small constant. The case $\tau_{\mathrm{rel}} = 1$ corresponds to i.i.d. distribution, and it follows from~\cite{jiao2017estimating} that $n^*(k, \tau_{\mathrm{rel}}, \epsilon) = \Theta(k)$ for small constant $\epsilon$. Interestingly, when the Markov chain becomes slightly ``non-i.i.d.'', the required sample size immediately jumps to $\Theta(k^2)$ and remains there no matter how large the $\tau_{\mathrm{rel}}$ is. Similar phenomena exist in the literature of entropy rate estimation for Markov chains~\cite{han2018entropy}, where the phase transitions for consistent entropy rate estimation happens at $\frac{k}{\log k}$ for i.i.d. data, and $\frac{k^2}{\log k}$ when the relaxation time is above $1 + \Omega \left( \frac{\log^2 k}{\sqrt{k}} \right)$. In other words, even if we use the codelength of the ``best'' universal code to estimate the entropy rate of the Markov source, it still requires considerably more samples than the information theoretically optimal entropy rate estimator that does not go through the construction of a code.

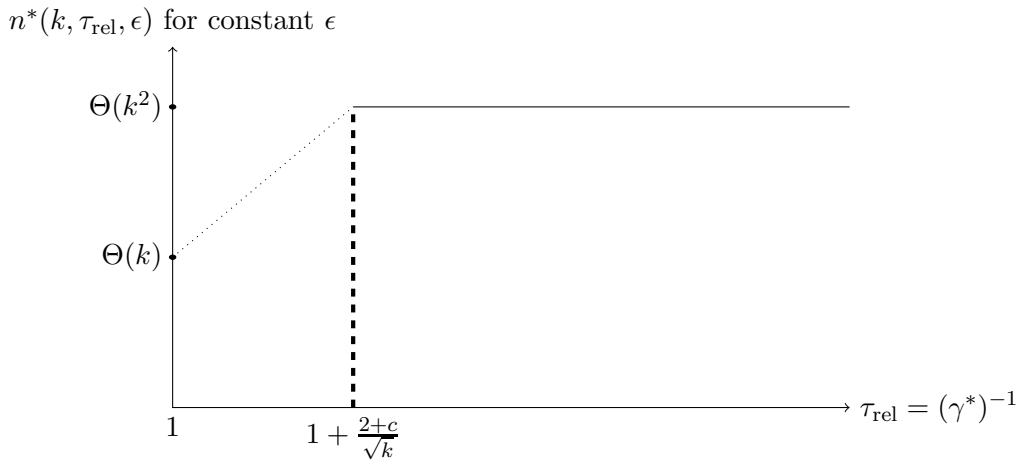
\begin{figure*}[htbp]
	\centering
	\begin{tikzpicture}[xscale=6, yscale=4]
	\draw [<->] (0,1.2) -- (0,0) -- (1.5,0);
	\draw [dashed, ultra thick] (0.4,0) -- (0.4,1.0);
	\draw [dotted] (0,0.5) -- (0.4,1.0);
	\draw (0.4, 1.0) -- (1.5, 1.0);

	\node [below] at (0,0) {1};
	\node [left] at (0,1.0) {$\Theta(k^2)$};
	\filldraw (0,1.0) circle (0.2pt);
	\node [left] at (0,0.5) {$\Theta(k)$};
	\filldraw (0,0.5) circle (0.2pt);
	\node [below] at (0.4,0) {$1+ \frac{2+c}{\sqrt{k}}$};
	\node [right] at (1.5,0) {$\tau_{\mathrm{rel}} = (\gamma^*)^{-1}$};
	\node [above] at (0,1.2) {$n^*(k, \tau_{\mathrm{rel}}, \epsilon)$ for constant $\epsilon$};
	\end{tikzpicture}
\captionof{figure}{The figure plots the $n^*(k, \tau_{\mathrm{rel}}, \epsilon)$ for a fixed small enough $\epsilon>0$, against the relaxation time constraint $\tau_{\mathrm{rel}}$. Note that $\tau_{\mathrm{rel}} = 1$ corresponds to i.i.d data, and hence 
$n^*(k, 1, \epsilon) = \Theta(k)$~\cite{jiao2017estimating}. }
	\label{fig.phasetransition}
	\end{figure*}

%
\section{Theorem \ref{thm.main} Proof Roadmap}\label{proof.thm.main}

We first conduct the continuous approximation of the redundancy, which is given by the following lemma. 

\begin{lemma}\label{lemma.contapproximation}
For any uniquely decodable code $L$, we have
\begin{align}
r_n(L,\theta)  & \geq \frac{1}{n} D(p_\theta(x^n) \| q_L(x^n)),
\end{align}
where $q_L(x^n) = \frac{2^{-L(x^n)}}{\sum_{x^n} 2^{-L(x^n)}}$. 
\end{lemma}

\begin{proof} Consider the redundancy $r_n(L,\theta)$:
\begin{align}
r_n(L,\theta) &= \frac{1}{n} \left[ E(L(X^n)|\theta) - H_{\theta}(X^n) \right]\\
			  &= \frac{1}{n} \left[ \sum_{x^n} p_{\theta}(x^n)L(x^n) - \sum_{x^n} p_{\theta}(x^n) \log \frac{1}{p_{\theta} (x^n)} \right]\\
              &=  \frac{1}{n} \left[  \sum_{x^n} p_{\theta}(x^n) \log \frac{p_{\theta} (x^n)\sum_{x^n} 2^{-L(x^n)}}{2^{-L(x^n)}} + \log \frac{1}{\sum_{x^n} 2^{-L(x^n)}}\right]
\end{align}
As $L$ is a uniquely decodable code \cite{cover2012elements} (section 5.1), we can now use the Kraft inequality \cite{cover2012elements} (Theorem 5.5.1) for the lengths $L(x^n)$ to obtain:
\begin{align}
r_n(L,\theta)   &\geq \frac{1}{n} \left[  \sum_{x^n} p_{\theta}(x^n) \log \frac{p_{\theta} (x^n)\sum_{x^n} 2^{-L(x^n)}}{2^{-L(x^n)}} \right]\label{eqn.kraft}\\
              &= \frac{1}{n} \left[  \sum_{x^n} p_{\theta}(x^n) \log \frac{p_{\theta} (x^n)}{q_L(x^n)}\right]\\
              &= \frac{1}{n} D ( p_{\theta}(x^n) || q_L(x^n)).
\end{align}
\end{proof}

We then use the strategy of lower bounding the minimax risk by Bayes risk, which is given by the following lemma. 
\begin{lemma}\label{lemma.Ibound} For any prior distribution $\Phi(\theta)$ supported on the parameter space $\mathcal{M}_{2,\mathrm{rev}}(k,\tau_{\mathrm{rel}})$, we have
\begin{equation}
R_n(k, \tau_{\mathrm{rel}}) \geq \frac{1}{n}I(\theta; X^n). 
\end{equation}
\end{lemma}
\begin{proof}
Let $p(x^n) = \int_{\theta} \Phi(\theta) p_{\theta}(x^n) d\theta$, then:
\begin{align}
R_n(k, \tau_{\mathrm{rel}}) &= \inf_{L} \sup_{K \in \mathcal{M}_{2,\mathrm{rev}}(k,\tau_{\mathrm{rel}})} r_n(L,\theta) \\
	&\geq \inf_{L} \int_{\theta} \Phi(\theta) r_n(L,\theta) d\theta \label{equation.prior}
\end{align}
Equation (\ref{equation.prior}) is true because the average is always lower than the supremum.
We next use the Lemma \ref{lemma.contapproximation} to obtain: 
\begin{align*}
 R_n(k, \tau_{\mathrm{rel}})  &\geq \inf_{q_L(x^n)} \frac{1}{n}\int_{\theta} \Phi(\theta) D(p_{\theta}(x^n) || q_L(x^n)) d\theta \\
    &= \inf_{q_L(x^n)} \frac{1}{n}\left[ \int_{\theta} \Phi(\theta) \sum_{x^n} p_{\theta}(x^n) \log \frac{p_{\theta}(x^n)} {q_L(x^n)} d\theta \right] \\
    &= \inf_{q_L(x^n)} \frac{1}{n}\left[ \int_{\theta} \Phi(\theta) \sum_{x^n} p_{\theta}(x^n) \log \frac{p_{\theta}(x^n) p(x^n)} {p(x^n)q_L(x^n)} d\theta \right] \\
    &= \inf_{q_L(x^n)} \frac{1}{n} \left[ \int_{\theta} \Phi(\theta) D(p_{\theta}(x^n) || p(x^n)) d\theta + D(p(x^n)||q_L(x^n)) \right]
\end{align*}
finally, the non-negativity of the KL-divergence, completes the proof:
\begin{align*}
 R_n(k, \tau_{\mathrm{rel}})  &\geq \frac{1}{n} \left[ \int_{\theta} \Phi(\theta) D(p_{\theta}(x^n) || p(x^n)) d\theta \right]\\
    &= \frac{1}{n}I(\theta;X^n)
\end{align*}
\end{proof}

Lemma~\ref{lemma.Ibound} suggests that, for any prior distribution $\Phi(\theta)$:
\begin{align}
R_n(k, \tau_{\mathrm{rel}}) &\geq \frac{1}{n}I(\theta;X^n) \\
							&= \frac{1}{n} [h(\theta) - h(\theta|X^n)].
\end{align}
In order to obtain a tight lower bound on $R_n(k,\tau_{\mathrm{rel}})$, it thus suffices to choose a prior on $\theta$ such that $h(\theta)$ is as large as possible, while $h(\theta|X^n)$, which quantifies how well we can estimate $\theta$ based on $X^n$, is as small as possible. 

The transition matrix has about $k^2$ degrees of freedom. In order to prove the lower bound corresponding to $n^*(k, \tau_{\mathrm{rel}}, \epsilon) \approx k^2$, we need nearly $k^2$ degrees of freedom in the prior construction, but would also like the Markov chain to mix fast under this prior. In other words, we want the Markov chain to be similar to the memoryless scenario. It naturally motivates a prior construction using random matrix theory. Indeed, if the transition matrix can be viewed as a combination of the rank one matrix corresponding to the stationary distribution and a ``noise'' matrix with nearly i.i.d. entries, it would be expected from random matrix theory that the second largest eigenvalue would be close to zero as the matrix size increases. However, the technical difficulty appears in constructing a prior which is \emph{completely} supported on $\mathcal{M}_{2,\mathrm{rev}}(k,\tau_{\mathrm{rel}})$ with desirable spectral properties and also in ensuring that the prior has large enough differential entropy. The concrete construction is below.

\subsection{Prior Construction}

Consider $\tilde{\mathcal{M}}_{2,\text{rev}}(k) \subset \mathcal{M}_{2,\text{rev}}(k)$ be the space of Markov distributions which have the following properties: 

\begin{equation}
\tilde{\mathcal{M}}_{2,\text{rev}}(k) = \{ K_{ij} \in \mathcal{M}_{2,\text{rev}}(k), K_{ii} = 0, \  \forall i \in [k] \} 
\end{equation}

The space $\tilde{\mathcal{M}}_{2,\text{rev}}(k)$ corresponds to transition matrices of random walks over undirected graphs that do not have self loops. We also define a class of stationary Markov chains $\tilde{\mathcal{M}}_{2,\text{rev}}(k, \tau_{\mathrm{rel}}) \subset \tilde{\mathcal{M}}_{2,\text{rev}}(k)$ as follows:
\begin{equation}
\tilde{\mathcal{M}}_{2,\text{rev}}(k, \tau_{\mathrm{rel}}) = \{ K_{ij} \in \tilde{\mathcal{M}}_{2,\text{rev}}(k), \tau_{\mathrm{rel}}(K) \leq \tau_{\mathrm{rel}}\}.
\end{equation}
In other words, we consider stationary reversible Markov chains in $\tilde{\mathcal{M}}_{2,\text{rev}}(k)$ whose relaxation time is upper-bounded by $\tau_{\mathrm{rel}}$. \\

\begin{definition}
Let $\pi(i,j) = \pi_i K_{ij}$ denote the stationary distribution over the tuples $(X_1,X_2)$. Then we can consider a parametrization $\theta$ for $\tilde{\mathcal{M}}_{2,\text{rev}}(k)$ as:
\begin{align*}
\theta &= (2\pi(1,2),\ldots,2\pi(1,k),2\pi(2,3),\ldots,2\pi(k-1,k)) \\
       &\equiv ({\theta}_{1,2},{\theta}_{1,3},\ldots,{\theta}_{1,k},{\theta}_{2,3},\ldots,{\theta}_{2,k},\ldots,{\theta}_{k-1,k})
\end{align*}
\end{definition}
The scaling by factor $2$ (e.g. $\theta_{1,2} = 2\pi(1,2)$) is considered to ensure that the sum of the parameters is $1$.
\begin{equation}
\sum_{i < j} \theta_{i,j} = 1
\end{equation}  
Note that, we can obtain the transition matrix $K$ from the parametrization $\theta$ as follows:\\
Let $\tilde{\theta}_{i,j}$ be defined as:
\begin{equation}
    \tilde{\theta}_{i,j}=
    \begin{cases}
      \theta_{i,j} &, \text{if}\ i<j \\
      0            &, \text{if}\ i=j \\
      \theta_{j,i} & ,\text{if}\ i>j
    \end{cases}
\end{equation}

Then, the transition matrix $K$ can be obtained as:
\begin{equation}
K_{ij} = \frac{\tilde{\theta}_{i,j}}{\sum_{j'} \tilde{\theta}_{i,j'}}
\end{equation}

We also define priors $\tilde{\Phi}^u(\theta)$ and $\tilde{\Phi}^u(\theta; {\tau_{\mathrm{rel}}})$ that are uniform distributions on spaces $\tilde{\mathcal{M}}_{2,\text{rev}}(k)$ and $\tilde{\mathcal{M}}_{2,\text{rev}}(k, \tau_{\mathrm{rel}})$, respectively, under the parametrization of $\theta$. We can obtain the distribution $\tilde{\Phi}^u(\theta)$ over the space $\tilde{\mathcal{M}}_{2,\text{rev}}(k)$, by considering transition matrices corresponding to random walk over undirected graphs with random weights. 

\begin{lemma}\label{lemma.exp}
Consider a simple complete graph (complete graph, without self loops) on $k$ vertices, with random weights $w_{ij}$ distributed i.i.d as $w_{ij} = w_{ji} \sim \mathsf{Exp}(1)$ ($w_{ii} = 0$). Then, the corresponding transition matrix $K$ is distributed as $\tilde{\Phi}^u(\theta)$, i.e. uniformly distributed over the  space $\tilde{\mathcal{M}}_{2,\text{rev}}(k)$.
\end{lemma}
\begin{proof}
We recall a well-known property\cite{frigyik2010introduction} (Section 2.3) of exponential distributions: Let $U = \{u_1,u_2,\ldots,u_r\}$ be such that every $u_i \sim $ i.i.d $\mathsf{Exp}(\lambda)$. Then, for $u=\sum_{i}u_i$ and $v_i = {u_i/u}$, the vector $V = \{ v_1,v_2,\ldots,v_r\}$ is uniformly distributed over the probability simplex $\sum_{i=1}^r v_i = 1$.

Lemma \ref{lemma.exp} is a special case, and can be proved by considering:
\begin{align*}
U &= \{2w_{12},2w_{13},\ldots,2w_{1k},2w_{23},\ldots,2w_{2k},\ldots,2w_{k-1,k}\} \\
V &= \theta \equiv \{{\theta}_{1,2},{\theta}_{1,3},\ldots,{\theta}_{1,k},{\theta}_{2,3},\ldots,{\theta}_{2,k},\ldots,{\theta}_{k-1,k}\}
\end{align*}
Then $V = \theta$ is uniformly distributed over $\tilde{\mathcal{M}}_{2,\text{rev}}(k)$, the probability simplex of dimension $\frac{k(k-1)}{2}$.

\end{proof}

Lemma \ref{lemma.exp} provides a nice gateway to use tools from random matrix theory. We will first understand some properties of the weight matrix $W$, consisting of random weights $w_{ij} = w_{ji} \sim \mathsf{Exp}(1),  \mbox{  } \forall i\neq j$ and $w_{ii} = 0, \mbox{  } \forall i \in [k]$. Recall the following definition of a Wigner's Matrix \cite{tao2012topics}

\begin{definition}
We say a random symmetric matrix $A$ is a Wigner's matrix, if the upper-triangular entries $A_{ij}, i > j$ are distributed i.i.d with zero mean and unit variance, while the diagonal entries $A_{ii}$ are i.i.d. real variables with bounded mean and variance, distributed independently of the upper-triangular entries.
\end{definition}

Consider the matrix $\hat{W}$, where $\hat{w}_{ij} = w_{ij}-1$. $\hat{W}$ is a symmetric random matrix, where the off-diagonal entries are i.i.d. with $0$ mean, while the diagonal entries are constants. This implies that, the matrix $\hat{W}$ is a Wigner's random matrix. Then, the Strong Bai-Yin theorem, upper bound (Theorem 2.3.24, Exercise 2.3.15 of \cite{tao2012topics}) implies that the eigenvalues of $\hat{W}$ are bounded as:
\begin{equation}\label{eqn.wigner}
|\lambda_i(\hat{W})| \leq 2 \sqrt{k} + o(\sqrt{k}) \mbox{  a.s.,   } \forall i \in [k]
\end{equation} 
(here, by a.s. we mean that the sequence of events is true infinitely often as $k \rightarrow \infty$)\\
We can also bound the row sums $\rho_i$ of the matrix $W$ as:
\begin{lemma}\label{lemma.rho}
The following properties are true for the weight matrix $W$. 
\begin{align}
\max_{1\leq i \leq k} \left| \frac{\rho_i}{k} - 1 \right| &= o(1) \mbox{ a.s.}\label{equation.1}\\
\sum_{1\leq i \leq k} \left( \frac{\rho_i}{k} - 1 \right)^2 &= O(1) \mbox{ a.s.}\label{equation.2}
\end{align}
\end{lemma}
\begin{proof}
Let $V$ be a matrix so that $v_{ij} = w_{ij}, \forall i \neq j$ and $v_{ii} \sim \mathsf{Exp}(1), \forall i$. Then, using Lemma 2.3 of \cite{bordenave2010spectrum}, we get that:
\begin{equation}
\max_{1\leq i \leq k} \left| \frac{\rho_i}{k} - 1 \right| \leq  o(1) + \frac{1}{k} \max_{1\leq i \leq k} v_{ii} \mbox{ a.s.}
\end{equation}
Let $A_{k,\epsilon}$ be the event such that: 
\begin{equation}
A_{k,\epsilon} = \left\{ \frac{1}{k} \max_{1\leq i \leq k} v_{ii} \leq \epsilon \right\} 
\end{equation} 
As $v_{ii}$ are independent exponential random variables, we obtain:
\begin{align*}
P(A^c_{k,\epsilon}) &= 1 - P\left( \frac{1}{k} \max_{1\leq i \leq k } v_{ii} \leq \epsilon \right) \\
				  &= 1 - \prod_{i=1}^k P( v_{ii} \leq k \epsilon )\\
                  &= 1 - (1 - e^{-k\epsilon} )^k\\
				  &\leq ke^{-k \epsilon}
\end{align*} 

We can now obtain a bound on the sum of the probability of events: 
\begin{align*}
\sum_{k=1}^{\infty} P(A^c_{k,\epsilon}) &\leq \sum_{k=1}^{\infty}ke^{-k\epsilon} \\
									 &< \infty
\end{align*}

By using Borel-Cantelli lemma, this proves the equation (\ref{equation.1}). 

From equation (\ref{eqn.wigner}), we know that largest eigenvalue of $\hat{W}$ is bounded as:
\begin{equation}
|\lambda_1(\hat{W})| \leq 2 \sqrt{k} + o(\sqrt{k}) \mbox{ a.s}
\end{equation} 

Thus:
\begin{align*}
\sum_{1\leq i \leq k} \left( \frac{\rho_i}{k} - 1 \right)^2 &= \frac{\langle\hat{W}\textbf{1},\hat{W}\textbf{1}\rangle}{k^2}\\
														   &\leq \frac{\lambda_1(\hat{W})^2}{k}\\
														   &\leq \frac{4k}{k} + o(1) \mbox{ a.s.}\\
														   &= O(1) \mbox{ a.s.}
\end{align*}
This proves the equation (\ref{equation.2}).
\end{proof}
The lemma \ref{lemma.rho} essentially says that, every $\rho_i$ is close to $k$, i.e.:
\begin{equation}\label{equation.bound}
\rho_i \leq k(1 + \delta), \mbox{ where } \delta = o(1) 
\end{equation}

Equation (\ref{equation.bound}) then implies that that entries of the matrix $K$, are proportional to that of the matrix $W$. 
\begin{equation}\label{equation.o1}
K_{ij} = \frac{w_{ij}}{\rho_i} = \frac{w_{ij}}{k} (1+o(1))
\end{equation}

This suggests that the eigenvalues of matrix $K$ are close to those of $W/k$. 

\begin{lemma}\label{lemma.spectrum}
Let $W$ be a random matrix on $k$ vertices, with random weights $w_{ij}$ distribted i.i.d as $w_{ij} = w_{ji} \sim \mathsf{Exp}(1)$ ($w_{ii} = 0$). Then the corresponding transition matrix $K$ has the following spectral properties:
\begin{align}
\lambda_1(K) &= 1 \\
\max_{2 \leq i \leq k} |\lambda_i(K)| &\leq  \frac{2+c}{\sqrt{k}}  \mbox{  a.s}
\end{align}
for some constant $c >0$. (here, by "a.s." we mean that the sequence of events is true infinitely often as $k \rightarrow \infty$)
\end{lemma}

The proof for the lemma \ref{lemma.spectrum} follows from lemma \ref{lemma.rho}, and is proved in the Appendix \ref{appendix.randommatrix}. The following corollary is immediate. 

\begin{corollary}\label{corollary.trel}
Let  $\theta \in \tilde{\mathcal{M}}_{2,\text{rev}}(k)$ be distributed according to the prior $\tilde{\Phi}^u(\theta)$. Also, let $\tau^0_{\mathrm{rel}} = 1+ \frac{2+c}{\sqrt{k}}$, where $c>0$ is a positive constant. Then, 
\begin{equation}
P(\theta \in \tilde{\mathcal{M}}_{2,\text{rev}}(k, \tau^0_{\mathrm{rel}})) \rightarrow 1,
\end{equation}
as $k\to \infty$. 
\end{corollary}


From now on we denote $\tau^0_{\mathrm{rel}} = 1+ \frac{2+c}{\sqrt{k}}$. We next analyze $h(\theta)$ and $h(\theta|X^n)$ under the prior $\tilde{\Phi}^u(\theta; {\tau^0_{\mathrm{rel}}})$. 

\begin{lemma}\label{lemma.h_theta}
Let $\theta \sim \tilde{\Phi}^u(\theta; {\tau^0_{\mathrm{rel}}})$ and $\tau^0_{\mathrm{rel}} = 1+ \frac{2+c}{\sqrt{k}}$. Then, the differential entropy $h(\theta)$ is lower bounded as
\begin{equation*}
h(\theta) \geq \frac{k(k-1)}{2} \log \frac{2}{k(k-1)} + \frac{k(k-1) \log e}{2} - \log k
\end{equation*}
for $k\geq k_c$, where $k_c$ only depends on $c>0$. 
\end{lemma}
\begin{proof}
Corollary~\ref{corollary.trel} implies that there exists some $k_c$ such that for $k\geq k_c$, 
\begin{align*}
 \mathsf{Vol}(\tilde{\mathcal{M}}_{2,\text{rev}}(k, \tau^0_{\mathrm{rel}})) &\geq  \frac{1}{2}\mathsf{Vol}(\tilde{\mathcal{M}}_{2,\text{rev}}(k)) \\
 &= \frac{1}{2}\left( \frac{1}{\left[\frac{k(k-1)}{2}\right]!}\right)
\end{align*}

As the distribution $\tilde{\Phi}^u(\theta; {\tau^0_{\mathrm{rel}}})$ is uniform, we know
\begin{align*}
h(\theta) &= \log \mathsf{Vol}(\tilde{\mathcal{M}}_{2,\text{rev}}(k, \tau^0_{\mathrm{rel}}))\\
 &\geq \log \mathsf{Vol}(\tilde{\mathcal{M}}_{2,\text{rev}}(k)) - 1 \\
  &\geq \frac{k(k-1)}{2} \log \frac{2}{k(k-1)} + \frac{k(k-1) \log e}{2} - {\log k}
\end{align*}

We used Stirling approximation for factorial to simplify the bound on the entropy. 
\end{proof}

The next step in the proof is to upper bound the term $h(\theta|X^n)$, which quantifies how well we can estimate the parameter $\theta$ from $X^n$. Let $\hat{\theta} = \theta(X^n)$ be a deterministic estimator for the parameter $\theta$. Then, 
\begin{align}
h(\theta|X^n) &= h(\theta - \hat{\theta} | X^n) \\
			  &\leq h(\theta - \hat{\theta})
\end{align}

Utilizing the fact that Gaussian distribution maximizes the differential entropy under variance constraints, we have
\begin{align}
	h(\theta - \hat{\theta}) &\leq \sum_{i,j} h(\theta_{i,j} - \hat{\theta}_{i,j}) \\
  							&\leq \frac{1}{2}\sum_{i,j}  \left[ \log \left(2\pi e \mathsf{\mathsf{Var}}(\hat{\theta}_{i,j}) \right)  \right] \\
    							&\leq \frac{1}{2}\sum_{i,j}  \left[ \log \left(2\pi e \int_{\theta} \tilde{\Phi}^u(\theta; {\tau^0_{\mathrm{rel}}}) \mathsf{\mathsf{Var}}(\hat{\theta}_{i,j} | \theta) d\theta \right) \right]  \label{equation.eq10}
\end{align}

Let $N(i,j) = \sum_{r=1}^{n-1} \mathbb{1}[(X_r,X_{r+1}) = (i,j)]$ represent the number of occurrences of the tuple $(i,j)$ in the $X^n$ sequence. Then, a natural estimator for parameter $\theta_{i,j} = 2\pi(i,j) = \pi(i,j) + \pi(j,i)$ is the empirical estimator $\hat{\theta}_{i,j}$:
\begin{equation}
\hat{\theta}_{i,j} = \frac{N(i,j) + N(j,i)}{n-1}
\end{equation}

We prove the following bound on the variance of the empirical estimator $\hat{\theta}_{i,j}$.

\begin{lemma} \label{lemma.var}
Let $\theta \sim \tilde{\Phi}^u(\theta; {\tau^0_{\mathrm{rel}}})$.  Then the variance of the estimator $\hat{\theta}_{i,j} = \frac{N(i,j) + N(j,i)}{n-1}$ can be bounded as:

\begin{equation}
\mathsf{Var}(\hat{\theta}_{i,j} | \theta ) \leq \frac{8\theta_{i,j} \tau^0_{\mathrm{rel}} }{n-1}
\end{equation}
\end{lemma}
\begin{proof}
Let $X_1 \rightarrow X_2 \ldots \rightarrow X_n$ be a reversible Markov chain with transition matrix $K \in \tilde{\mathcal{M}}_{2,\text{rev}}(k, \tau^0_{\mathrm{rel}})$. Consider the Markov chain over the tuples $(X_1,X_2) \rightarrow (X_2,X_3) \ldots \rightarrow (X_{n-1},X_n)$. Let $\tilde{K}$ be the corresponding transition matrix. It is interesting to note that, although the original Markov chain is reversible, the chain over the tuples is generally not reversible. Let:
\begin{equation*}
f_{i,j}(X_{r-1},X_r) = \frac{\mathbb{1}[(X_{r-1},X_r)= (i,j)] + \mathbb{1}[(X_{r-1},X_r)= (j,i)]}{n-1}
\end{equation*}

then, the estimator $\hat{\theta}_{i,j}$ can be written as:
\begin{equation*}
\hat{\theta}_{i,j} = \sum_{r=1}^{n-1} f_{i,j}(X_{r-1},X_r) 
\end{equation*}

We can now use Theorem 3.7 of \cite{paulin2015concentration} on the function $f_{i,j}(X_{r-1},X_r)$ corresponding to the Markov chain over the tuples $(X_1,X_2) \rightarrow (X_2,X_3) \ldots \rightarrow (X_{n-1},X_n)$, to obtain:  
\begin{align*}
\mathsf{Var}(\hat{\theta}_{i,j} | \theta ) &\leq \frac{4 \theta_{i,j}}{\gamma_{ps}(\tilde{K}) (n-1)}
\end{align*}

We next use the lemma \ref{lemma.gammaps} (proved in the Appendix \ref{appendix.varproof}) to bound $\gamma_{ps}(\tilde{K})$ in terms of $\tau_{\mathrm{rel}}(K)$:
\begin{lemma}\label{lemma.gammaps}
Let $X_1 \rightarrow X_2 \ldots \rightarrow X_n$ be a reversible Markov chain with transition matrix $K$, then the Markov chain over tuples, $(X_1,X_2) \rightarrow (X_2,X_3) \ldots \rightarrow (X_{n-1},X_n)$ has pseudo-spectral gap $\gamma_{ps}(\tilde{K})$ given by:
\begin{equation}
\gamma_{ps}(\tilde{K}) \geq \frac{\gamma^*(K)}{2}
\end{equation}
\end{lemma}

Using lemma \ref{lemma.gammaps}, we obtain the variance bound:

\begin{align*}
\mathsf{Var}(\hat{\theta}_{i,j} | \theta ) &\leq \frac{8 \theta_{i,j}}{\gamma^*(K) (n-1)}\\
										  &= \frac{8 \theta_{i,j} \tau_{rel}(K)}{(n-1)} \\
										  &\leq \frac{8 \theta_{i,j} \tau^0_{rel}(K)}{(n-1)} 
\end{align*}
which proves the lemma.
\end{proof}

We next use the variance bound on the estimator to obtain an upper bound on $h(\theta|X^n)$.

\begin{lemma}\label{lemma.h_theta_xn}
Let $\theta \sim \tilde{\Phi}^u(\theta; {\tau^0_{\mathrm{rel}}})$, then the conditional differential entropy $h(\theta|X^n)$ is upper bounded by:

\begin{equation}
h(\theta|X^n) \leq \frac{k(k-1)}{4} \log \frac{16\pi e {\tau^0_{\mathrm{rel}}} }{n-1} + \frac{k(k-1)}{4}\log \frac{2}{k(k-1)}	
\end{equation}
\end{lemma}
\begin{proof}

From equation (\ref{equation.eq10}), and lemma \ref{lemma.gammaps} we obtain:

\begin{align}
	h(\theta|X^n) &\leq \frac{1}{2}\sum_{i,j}  \left[ \log \left(2\pi e \int_{\theta} \tilde{\Phi}^u(\theta; {\tau^0_{\mathrm{rel}}})\mathsf{Var}(\hat{\theta}_{i,j} | \theta) d\theta \right)  \right] \\
				  &\leq \frac{1}{2}\sum_{i,j}  \left[ \log \left(2\pi e \int_{\theta} \tilde{\Phi}^u(\theta; {\tau^0_{\mathrm{rel}}}) \frac{8\theta_{i,j} \tau^0_{\mathrm{rel}} }{n-1} d\theta  \right)  \right] \\
				  &= \frac{k(k-1)}{4} \log \frac{16\pi e {\tau^0_{\mathrm{rel}}} }{n-1} + \frac{1}{2} \sum_{i,j}  \left[ \log \int_{\theta} \tilde{\Phi}^u(\theta; {\tau^0_{\mathrm{rel}}}) \theta_{i,j} d\theta \right]  \\
				  &\leq \frac{k(k-1)}{4} \log \frac{16\pi e {\tau^0_{\mathrm{rel}}} }{n-1} \\ 
				  &\quad+ \frac{k(k-1)}{4} \left[ \log \sum_{i,j} \frac{2}{k(k-1)} \int_{\theta} \tilde{\Phi}^u(\theta; {\tau^0_{\mathrm{rel}}}) \theta_{i,j} d\theta  \right] \\
				 &= \frac{k(k-1)}{4} \log \frac{16\pi e {\tau^0_{\mathrm{rel}}} }{n-1} + \frac{k(k-1)}{4}\log \frac{2}{k(k-1)}	\label{equation.65}				  
\end{align}
Equation (\ref{equation.65}) is true because $\sum_{i,j} \theta_{i,j} = 1$. This proves the upper bound on the term $h(\theta|X^n)$.
\end{proof}

\subsection{Proof of Theorem \ref{thm.main}}

Using lemma~\ref{lemma.Ibound} for the prior $\tilde{\Phi}^u(\theta; {\tau^0_{\mathrm{rel}}})$, and lemma~\ref{lemma.h_theta},~\ref{lemma.h_theta_xn}, we have
\begin{align}
R_n(k, \tau_{\mathrm{rel}}) &\geq \frac{1}{n}I(\theta;X^n) \\
							&= \frac{1}{n} [h(\theta) - h(\theta|X^n)] \\
							&\geq \frac{k(k-1)}{4n}\log \frac{2(n-1)}{k(k-1)} + \frac{k(k-1)}{4n} \log \frac{e}{16\pi {\tau^0_{\mathrm{rel}}}} - \frac{\log k}{n}				
\end{align}
for $k\geq k_c$, where $\tau_{\mathrm{rel}}^0 = 1 + \frac{2+c}{\sqrt{k}}$. 
This completes the proof of the lower bound.

\section{Theorem \ref{thm.main2} Proof}\label{proof.thm.main2}
For any sequence $x^n$ over the alphabet $\mathcal{X} = [k]$, let $N(a), N(a,b)$ be defined as:
\begin{align}
N(a) &= \sum_{i=1}^{n-1} \mathbb{1}[x_i = a]\\
N(a,b) &= \sum_{i=1}^{n-1} \mathbb{1}[(x_i,x_{i+1}) = (a,b)]
\end{align}

Before we prove the theorem, we consider some simple lemmas.
\begin{lemma}\label{lemma.universalprefix}
There exists a prefix code \cite{cover2012elements} (Section 5.1) on non-negative integers $\mathbb{N} \cup \{0\} = \{0,1,2,\ldots\}$, such that every integer $m$ has a codeword of length $l_m \leq 2\log_2 (m+1) + 1$.
\end{lemma}
\begin{proof}
Let $q$ be such that: $2^q \leq (m+1) < 2^{q+1}$. Thus, $(m+1)$ can be written as:
\begin{equation}
(m+1) = 2^q + r
\end{equation}
where, $0 \leq r < 2^q$. Let $U_q = 000\ldots001$ be a unary code with $q$ zeros, and $B_r$ be the binary representation of $r$ using $q$ bits. Then the following code $C_m$ is prefix-free:

\begin{equation}
C_m = U_q 1 B_r
\end{equation}
Thus, the length of the code $l(C_m)$:
\begin{align}
l(C_m) &= 2q + 1 \\
       &\leq 2\log_2 (m+1) + 1
\end{align} This completes the proof.
\end{proof}

\begin{lemma} We can store the parameters $N(a),N(a,b), \forall a,b \in [k]$ for a sequence $x^n$ using $L_{param}$ number of bits, which is upper bounded as:

\begin{equation}
L_{param}(x^n) \leq 2k^2\log_2 \left( \frac{n}{k^2}  + 1\right) + k^2
\end{equation}
\end{lemma}
\begin{proof}
Note that, we only need to store $N(a,b), \forall a,b \in [k]$ as the parameters $N(a)$ can be derived. Using the prefix coding from lemma \ref{lemma.universalprefix} for the parameters $N(a,b)$:
\begin{align}
L_{param}(x^n) &\leq \sum_{a,b} [2\log_2(N(a,b) + 1) + 1] \\
          &= 2k^2\sum_{a,b} \frac{1}{k^2}\log_2(N(a,b) + 1) + k^2\\
          &\leq 2k^2\log_2 \left( \frac{1}{k^2} \sum_{a,b} (N(a,b) + 1)\right) + k^2\label{eqn.concavitylog}\\
          &= 2k^2\log_2 \left( \frac{n}{k^2}  + 1\right) + k^2
\end{align}
Equation (\ref{eqn.concavitylog}) is true due to concavity of the $\log$ function and the Jensen's inequality.
\end{proof}

\begin{lemma}\label{lemma.arithmetic_coding} We can use arithmetic coding \cite{witten1987arithmetic} to encode a sequence $x^n$ using $L_{seq}(x^n)$ bits, which is bounded as:
\begin{equation}
L_{seq}(x^n) \leq \log_2 k + (n-1)H_1(x^n) + 3
\end{equation}
where $H_1(x^n)$ is the $1^{st}$ order empirical entropy of sequence $x^n$:
\begin{align}
H_1(x^n) &= \sum_{a=1}^k \sum_{b=1}^k \frac{N(a,b)}{n-1} \log_2 \frac{N(a)}{N(a,b)}
\end{align}
\end{lemma}
\begin{proof}
We first encode the sequence $x_1$ using fixed $\lceil\log_2 k \rceil $ bits. Next, we encode the remaining $(n-1)$ symbols  using arithmetic coding \cite{willems1995context} (section IV) with the first order model distribution $q(b|a) = N(a,b)/N(a)$. Using theorem 1 of  \cite{willems1995context}, the codelength of $x^n$ is:
\begin{align}
L_{seq}(x^n) &\leq \lceil \log_2 k \rceil + \left(\sum_{i=1}^{n-1} \log_2 \frac{1}{q(x_{i+1}|x_i)} + 2 \right) \label{eqn.arithmetic}\\
        &= \log_2 k + 1 + \sum_{a=1}^k \sum_{b=1}^k N(a,b) \log_2 \frac{N(a)}{N(a,b)} + 2\\
        &= \log_2 k + (n-1)H_1(x^n) + 3     
\end{align}
This completes the proof.
\end{proof}

Let $x^n$ be a given sequence over the alphabet $\mathcal{X} = [k]$. Consider the following compressor:
\begin{enumerate}
\item Store all the parameters $N(a,b), \forall a,b \in [k]$ using the universal prefix-free code in lemma \ref{lemma.universalprefix}.
\item Use the parameters $N(a,b)$ to compress $x^n$ using first-order Markov arithmetic coding as in lemma  \ref{lemma.arithmetic_coding}.
\end{enumerate}

Then, the codelength $\hat{L}(x^n)$ is bounded as:

\begin{align}
\hat{L}(x^n) &= L_{param}(x^n) + L_{seq}(x^n)\\
             &\leq \left[ 2k^2\log_2 \left( \frac{n}{k^2}  + 1\right) + k^2 \right] + [ \log_2 k + (n-1)H_1(x^n) + 3 ]
\end{align}

%

We now take a look at redundancy $R_n(k)$:

\begin{align}
R_n(k) &= \inf_{L} \sup_{\theta \in \mathcal{M}_{2}(k)} r_n(L,\theta) \\
    &\leq \sup_{\theta \in \mathcal{M}_{2}(k)} \frac{1}{n} \left( \mathbb{E}_{\theta}[\hat{L}(X^n)] - H_{\theta}(X^n)) \right)\\
    &= \sup_{\theta \in \mathcal{M}_{2}(k)} \frac{1}{n} \left( \mathbb{E}_{\theta}[\hat{L}(X^n)] - H_{\theta}(X_1) - (n-1)H_{\theta}(X_2|X_1) \right)\\
    &\leq \sup_{\theta \in \mathcal{M}_{2}(k)} \frac{1}{n} \left( \mathbb{E}_{\theta}[\hat{L}(X^n)] - (n-1)H_{\theta}(X_2|X_1) \right)\\
    &\leq \frac{2k^2}{n}\log_2 \left( \frac{n}{k^2}  + 1\right) + \frac{k^2}{n} 
    + \sup_{\theta \in \mathcal{M}_{2}(k)} \frac{n-1}{n} \left( \mathbb{E}_{\theta}[H_1(X^n)] - H_{\theta}(X_2|X_1) \right) + \frac{\log_2 k + 3}{n}\\
    &\leq \frac{2k^2}{n}\log_2 \left( \frac{n}{k^2}  + 1\right) + \frac{k^2}{n} + \frac{\log_2 k + 3}{n} \label{equation.concavity}
\end{align}

Where equation (\ref{equation.concavity}) is true because of the concavity of entropy. This completes the proof of the upper bound.

\bibliographystyle{alpha}
\bibliography{di}

\appendix
\section{Existing Minimax Redundancy Lower Bounds} \label{appendix.davisson}
We analyze the existing lower bound by \cite{davisson1983minimax}. 
\begin{align}
R_n(k) \geq g(k,n) &= \frac{k(k-1)}{2n}\log n + \frac{k(k-1)}{n} \log \frac{1}{k^4}- \frac{k(k-1)}{2n}\log \left[ \frac{C}{1-\left( 1 - \frac{1}{4k^4}\right)^{\frac{1}{2}}}\right] 
\end{align}

We can simplify $g(k,n)$, to get the lower bound:
\begin{align}
g(k,n) &= \frac{k(k-1)}{2n}\log n + \frac{k(k-1)}{n} \log \frac{1}{k^4} - \frac{k(k-1)}{2n}\log \left[ \frac{C}{1-\left( 1 - \frac{1}{4k^4}\right)^{\frac{1}{2}}}\right] \\
    &= \frac{k(k-1)}{2n}\log \frac{n}{k^2} - \frac{5k(k-1)}{2n} \log k^2 - \frac{k(k-1)}{2n} \log C + o\left(\frac{k(k-1)}{n}\right) 
\end{align}

Thus, the effective lower bound on $R_n(k)$ is:

\begin{align}
R_n(k) &\geq \frac{k(k-1)}{2n}\log \frac{n}{k^2} - \frac{5k(k-1)}{2n} \log k^2 - \frac{k(k-1)}{2n} \log C + \left(\frac{k(k-1)}{n}\right)
\end{align}

We observe that the lower bound on redundancy $R_n(k)$ is non-zero only when $n \gg k^2\log k$. We aim to improve the lower bound when $n \asymp k^2$.

\section{Proof of Lemma \ref{lemma.spectrum}} \label{appendix.randommatrix} 

To analyze the spectrum of the transition matrix $K$, we construct a symmetric matrix $S$, which has the same spectrum as $K$ almost surely.
\begin{lemma}\label{lemma.spectral_equivalence}
(Spectral Equivalence) Almost surely, for a large $k$, the spectrum of the transition matrix $K$ coincides with the spectrum of the symmetric matrix $S$ defined as: 

\begin{equation}
S_{ij} = \sqrt{\frac{\rho_i}{\rho_j}}K_{ij} = \frac{w_{ij}}{\sqrt{\rho_i \rho_j}}
\end{equation}
\end{lemma}
The lemma is proved in lemma 2.1 \cite{bordenave2010spectrum}. 

We now use the lemma \ref{lemma.rho}, which allows us to estimate $\rho_i = k(1+ o(1))$, to compare the spectrum of the matrix $\sqrt{k}K$ with the matrix $\frac{W}{\sqrt{k}}$.

\begin{lemma}\label{lemma.bulkbehavior}
(Bulk behavior) The ESD (empirical spectral density) of $\sqrt{k}K$ weakly converges to the Wigner's semi-circle law $\mathcal{W}_2$. 
\begin{equation}
\mu_{\sqrt{k}K} \xrightarrow[k \rightarrow \infty]{w} \mathcal{W}_{2}
\end{equation}

where the Wigner's semi-circle law $\mathcal{W}_2$ is given by:

\begin{equation*}
x \mapsto \frac{1}{2\pi} \sqrt{4 - x^2} \mathbb{1}_{[-2,2]}(x)
\end{equation*}
\end{lemma}
\begin{proof}
First of all, from the lemma \ref{lemma.spectral_equivalence}, the spectrum of $S$ is equivalent to that of $K$ a.s. (for large $k$). Thus, it is sufficient to analyze the spectrum of $\sqrt{k}S$. To show the weak convergence, we bound the Levy distance between the cumulative distributions corresponding to the ESD of matrices $\sqrt{k}S$ and $W/\sqrt{k}$. Let $F_{\sqrt{k}S}$ and $F_{W/{\sqrt{k}}}$ be the cumulative distributions, then:

\begin{align}
L^3(F_{\sqrt{k}S},F_{W/{\sqrt{k}}}) &\leq \frac{1}{k} Tr((\sqrt{k}S - {W/{\sqrt{k}}})^2) \\
								   &= \frac{1}{k} \sum_{i,j}^k \frac{w_{ij}^2}{k} \left( \frac{k}{\sqrt{\rho_i \rho_j}} - 1 \right)^2\\
								   &\leq O(\delta^2) \left( \frac{1}{k^2} \sum_{i,j}^k w_{ij}^2 \right)\\
								   &\rightarrow 2O(\delta^2) \mbox{ a.s, as } k \rightarrow \infty 
\end{align}
This proves the weak convergence of the $\mu_{\sqrt{k}K}$ to the wigner semi-circle law $\mathcal{W}_{2}$.
\end{proof} 

Note that, even though $\lambda_1(\sqrt{k}K) = \sqrt{k} \rightarrow \infty$ as $k \rightarrow \infty$, the weak limit of of $\mu_{\sqrt{k}K}$ is not affected since $\lambda_1(\sqrt{k}K)$ has weight $1/k$. The theorem thus implies that the bulk of the spectrum $\sigma(K)$ collapses as $k^{-1/2}$, but does not give a characterization for $\lambda_2(\sqrt{k}K)$, which is what is required.

To prove lemma \ref{lemma.spectrum}, we represent the symmetric matrix $S$ as a combination of a rank one matrix $P$ corresponding to the stationary distribution, and a "noise" matrix $S-P$ with nearly i.i.d entries. Bounding the spectral norm of the "noise" matrix $S-P$ gives us the result.

\begin{proof}(lemma \ref{lemma.spectrum})
Since $K$ is almost surely irreducible, for large enough $k$, the eigenspace of $S$ of eigenvalue $1$ is a.s. of size 1. and is the span of the vector $[\sqrt{\rho_1},\sqrt{\rho_2},\ldots,\sqrt{\rho_k}]$. Consider the symmetric matrix $P$: 

\begin{equation}\label{equation.80}
P_{ij} = \frac{\sqrt{\rho_i\rho_j}}{\rho}
\end{equation}

By removing $P$ from $S$, we are essentially removing the largest eigenvalue $1$, without touching the other eigenvalues. Thus, the spectrum of the matrix $S-P$ is given by:
\begin{equation}
\{ \lambda_2(S),\lambda_3(S),\ldots,\lambda_k(S)\} \cup \{0\}
\end{equation}

To find $\sqrt{k} \lambda_2(S)$, we now bound the spectral norm of matrix $A = \sqrt{k} (S-P)$. Lemma 2.4 of \cite{bordenave2010spectrum} along with lemma \ref{lemma.rho} gives us the result: 

\begin{equation}\label{equation.81}
\max_{1\leq i \leq k} \sqrt{k}|\lambda_i(S-P)| \leq 2 + o(1) \mbox{ a.s.}
\end{equation}

Equation (\ref{equation.80}) and lemma \ref{lemma.spectral_equivalence} together imply that:
\begin{align}
\max_{2\leq i \leq k} \sqrt{k}|\lambda_i(S)| &\leq 2 + o(1) \mbox{ a.s.}\\
\max_{2\leq i \leq k} \sqrt{k}|\lambda_i(K)| &\leq 2 + o(1) \mbox{ a.s.}\\
\max_{2\leq i \leq k} |\lambda_i(K)| &\leq \frac{2+c}{\sqrt{k}} \mbox{ a.s.}
\end{align}
for some constant $c\ge 0$. This completes the proof.

\end{proof}

\section{Proof of Lemma \ref{lemma.gammaps}} \label{appendix.varproof}
Consider the Markov chain over the tuples $(X_1,X_2) \rightarrow (X_2,X_3) \ldots \rightarrow (X_{n-1},X_n)$ with transition matrix $\tilde{K}$. We first analyze some properties of the transition matrix $\tilde{K}$.

\begin{lemma}\label{lemma.Ktilda}
Let $X_1 \rightarrow X_2 \ldots \rightarrow X_n$ be a reversible Markov chain with transition matrix $K$, then the transition matrix $\tilde{K}$ for the chain: $(X_1,X_2) \rightarrow (X_2,X_3) \ldots \rightarrow (X_{n-1},X_n)$ and its matrix reversibilization $\tilde{K}^{*}$ are:
\begin{align*}
\tilde{K}((a,b),(c,d)) &= \mathbb{1}[b=c]K(c,d) \\
\tilde{K}^{*}((a,b),(c,d)) &= \mathbb{1}[a=d]K(d,c) 
\end{align*}
\end{lemma}
\begin{proof}
The transition matrix for the chain $(X_1,X_2) \rightarrow (X_2,X_3) \ldots \rightarrow (X_{n-1},X_n)$ is given by:
\begin{align}
\tilde{K}((a,b),(c,d)) &= P(X_2 = c, X_3 = d | X_1 = a, X_2 = b) \\
                       &= P(X_2 = c | X_1 = a, X_2 = b) P( X_3 = d | X_1 = a, X_2 = b, X_2 = c)\\
                       &= \mathbb{1}[b=c]P( X_3 = d | X_2 = c) \label{eqn.markov}\\
                       &= \mathbb{1}[b=c]K(c,d)
\end{align}
where, equation(\ref{eqn.markov}) holds because of the Markovity condition.

As the Markov chain $(X_1,X_2) \rightarrow (X_2,X_3) \ldots \rightarrow (X_{n-1},X_n)$ is in general non-reversible, the multiplicative reversibilization of $\tilde{K}$ is:
\begin{align}
\tilde{K}^{*}((c,d),(a,b)) &= \tilde{K}((a,b),(c,d))\frac{\pi(a,b)}{\pi(c,d)} \\
			   &= \mathbb{1}[b=c]K(c,d)\frac{\pi(a)K(a,b)}{\pi(c)K(c,d)} \\
               &= \mathbb{1}[b=c]\frac{\pi(b)K(b,a)}{\pi(c)} \\
               &= \mathbb{1}[b=c]\frac{\pi(b)K(b,a)}{\pi(b)} \\
\tilde{K}^{*}((c,d),(a,b)) &= \mathbb{1}[b=c]K(b,a) \\
\tilde{K}^{*}((a,b),(c,d)) &= \mathbb{1}[a=d]K(d,c)              
\end{align}

This proves the lemma.
\end{proof}

\begin{lemma}\label{lemma.T}
Let $\mathcal{T}$ be the $k \times k$ matrix corresponding to the transformation: 
\begin{align}
\mathcal{T} M((a,b),(c,d)) &= M((b,a),(c,d))
\end{align}
 
Then, $\tilde{K}$ and $\tilde{K}^{*}$ have the property:
\begin{align}
((\tilde{K}^{*})^r\tilde{K}^r) &= (\mathcal{T} \tilde{K}^r)^2
\end{align}
\end{lemma}
\begin{proof}
The matrix $\mathcal{T}$ also has the properties:
\begin{align}
 M((a,b),(c,d)) \mathcal{T} &= M((a,b),(d,c)) \\
\mathcal{T}^2 &= I
\end{align}

Then we can show that:
\begin{align}
 \mathcal{T} \tilde{K}((a,b),(c,d)) \mathcal{T} &= \mathcal{T} \mathbb{1}[b=c]K(c,d) \mathcal{T} \\
 			&=  \mathbb{1}[a=c]K(c,d) \mathcal{T} \\
            &=  \mathbb{1}[a=d]K(d,c) \\
            &= \tilde{K}^{*}((a,b),(c,d))\label{eq1.Kstar}
\end{align}

Using equation (\ref{eq1.Kstar}) we can show that, for any $r$:
\begin{align}
((\tilde{K}^{*})^r\tilde{K}^r) &= ( \mathcal{T} \tilde{K}  \mathcal{T} )^r \tilde{K}^r \\
			   &= \mathcal{T} \tilde{K}^r  \mathcal{T} \tilde{K}^r \\
               &= (\mathcal{T} \tilde{K}^r)^2
\end{align}

This completes the proof.
\end{proof}

\begin{lemma}\label{lemma.Ksquare_tksquare}
Matrices $\tilde{K}^2$ and $\mathcal{T} \tilde{K}^2$ have the form:
\begin{align}
\tilde{K}^2((a,b),(c,d)) &= K(b,c)K(c,d)\label{equation.83}\\
\mathcal{T} \tilde{K}^2((a,b),(c,d)) &= K(a,c)K(c,d)\label{equation.84}
\end{align}
\end{lemma}
\begin{proof}
Using lemma \ref{lemma.Ktilda}, we obtain:
\begin{align}
\tilde{K}^2((a,b),(c,d)) &= \sum_{e,f} \tilde{K}((a,b),(e,f)) \tilde{K}((e,f),(c,d)) \\
                         &= \sum_{e,f} \mathbb{1}[b=e]K(e,f)\mathbb{1}[f=c]K(c,d)\\
                         &= K(b,c)K(c,d) 
\end{align}
This proves the equation (\ref{equation.83}). We now use the definition of $\mathcal{T}$ to obtain:

\begin{align}
\mathcal{T}\tilde{K}^2((a,b),(c,d)) &= \tilde{K}^2((b,c),(c,d)) \\
                         &= K(a,c)K(c,d) 
\end{align}
This proves the lemma.
\end{proof}

\begin{lemma}\label{lemma.KtKsq}
Matrices $K$ and $\mathcal{T}\tilde{K}^2$ have identical non-zero eigenvalues.
\end{lemma}
\begin{proof}

Let $V$ be an eigenvector of the matrix $\mathcal{T} \tilde{K}^2$ with a non-zero eigenvalue $\eta$. This implies that:
\begin{align}
\eta V((a,b),1) &= \sum_{c,d} \mathcal{T} \tilde{K}^2((a,b),(c,d))V((a,b),1)\\
			    &= \sum_{c,d} K(a,c)K(c,d) V((a,b),1) \\
			    &= \sum_{c} K(a,c)V((a,b),1) \sum_d K(c,d) \\
			    &= \sum_{c} K(a,c)V((a,b),1) 
\end{align}

Thus, this shows that for any $b \in [k]$, the vector $V((.,b),1)$ is an eigenvector of matrix $K$ with eigenvalue $\eta$. 

Conversely, let $v = [v_1,v_2,\ldots,v_k]^T$ be an eigenvector of the matrix $K$ with non-zero eigenvalue $\eta$. Then, the vector $V((a,b),1) = v_a$ is an eigenvector of the matrix $\mathcal{T} \tilde{K}^2$. Thus, together this implies that the non-zero eigenvalues of matrices $K$ and $\mathcal{T} \tilde{K}^2$ are identical. 
%


\end{proof}

We now come to the proof of lemma \ref{lemma.gammaps}.
\begin{proof} (lemma \ref{lemma.gammaps})
Using lemma \ref{lemma.KtKsq} and lemma \ref{lemma.T} we get:
\begin{align*}
\lambda_2((\mathcal{T} \tilde{K}^2)^2) &= \max(\lambda_2(K)^2,\lambda_k(K)^2)\\
\gamma((\tilde{K}^*)^2\tilde{K}^2) &= 1 - \max(\lambda_2(K)^2,\lambda_k(K)^2) \\
                                   &\geq 1 -  \max(|\lambda_2(K)|,|\lambda_k(K)|) \\
                                   &= \gamma^*(K) 
\end{align*}

Now using the definition of the pseudo-spectral gap $\gamma_{ps}(\tilde{K})$, we obtain:                                  
\begin{align*}                     
\gamma_{ps}(\tilde{K}) &\geq \frac{\gamma((\tilde{K}^*)^2\tilde{K}^2)}{2} \\
                       &\geq  \frac{\gamma^*(K)}{2} 
\end{align*}
This proves the lemma.
\end{proof}

\end{document}